\documentclass[twocolumn,amsthm]{autart}    
                                     
\newtheorem{theorem}{Theorem}
\newtheorem{lemma}{Lemma}

\newtheorem{assumption}{Assumption}
\newtheorem{proposition}{Proposition}
\newtheorem{definition}{Definition}

\journal{Automatica}

\usepackage{algorithm,algorithmic}
\usepackage{graphicx}          

\usepackage{etex}
\usepackage{physics}
\usepackage{color}
\usepackage{array}
\usepackage{multirow}
\usepackage{bigdelim}
\usepackage{ctable}
\usepackage{mdwtab}
\usepackage{cases}
\usepackage{soul}
\usepackage{epstopdf}
\usepackage{amssymb}
\usepackage{cite}
\usepackage{enumerate}

\newcommand{\x}[3][]{%
  \ensuremath{x_{[#2,#3]}^{#1}}%
}

\renewcommand{\u}[3][]{%
  \ensuremath{u_{[#2,#3]}^{#1}}%
}

\DeclareMathOperator*{\interior}{int}

\DeclareMathOperator*{\nullM}{nullity}

\usepackage{accents}

\begin{document}

\begin{frontmatter}

\title{Online Inverse Optimal Control for Control-Constrained Discrete-Time Systems on Finite and Infinite Horizons\thanksref{footnoteinfo}} 

\thanks[footnoteinfo]{%
Early versions of some results in this paper were presented at the IEEE CDC 2018 conference.\\
This work was supported by the Queensland Government Department of Science, Information Technology and Innovation (DSITI) and Boeing Research And Technology - Australia (BR\&T-A) through an Advance Queensland Research Fellowship to T.~L.~Molloy.\\
J.~Ford acknowledges continued support from the Queensland University of Technology (QUT) through the Centre for Robotic.\\
T.~L.~Molloy and T.~Perez were with the School of Electrical Engineering and Robotics, Queensland University of Technology (QUT), Brisbane, QLD 4000, Australia.\\
Corresponding author T.~L.~Molloy.}
\author[melb,qut]{Timothy L.~Molloy}\ead{tim.molloy@unimelb.edu.au},    
\author[qut]{Jason J.~Ford}\ead{j2.ford@qut.edu.au}, and   
\author[qut]{Tristan~Perez}

\address[melb]{Department of Electrical and Electronic Engineering, University of Melbourne, Parkville, VIC 3010, Australia}  
\address[qut]{School of Electrical Engineering and Robotics, Queensland University of Technology (QUT), Brisbane, QLD 4000, Australia} 
          

\begin{abstract} 
In this paper, we consider the problem of computing parameters of an objective function for a discrete-time optimal control problem from state and control trajectories with active control constraints.
We propose a novel method of inverse optimal control that has a computationally efficient online form in which pairs of states and controls from given state and control trajectories are processed sequentially without being stored or processed in batches.
We establish conditions guaranteeing the uniqueness of the objective-function parameters computed by our proposed method from trajectories with active control constraints.
We illustrate our proposed method in simulation.
\end{abstract}
\end{frontmatter}



\section{Introduction}
\label{section:introduction}
Many applications in control engineering \cite{Keshavarz2011,Panchea2017,Molloy2016,Molloy2018,Yokoyama2017,Maillot2013}, economics \cite{Konstantakopoulos2017}, and robotics \cite{Mombaur2010,Levine2012,Puydupin2012} involve inferring the underlying objectives of agents and systems from their behaviours.
Inverse optimal control (or inverse reinforcement learning) is a promising methodology for computing the objectives of control systems from given state and control trajectories, and its recent applications include learning driving styles \cite{Kuderer2015}, modelling human movement \cite{Mombaur2010}, and inferring the intent of aircraft \cite{Yokoyama2017}.
Motivated by applications involving systems subject to control constraints and operating indefinitely in real-time (such as vehicles), in this paper we propose a novel method of inverse optimal control for when the optimisation horizon is unknown, the controls are subject to constraints, and the given trajectories are to be processed recursively online.

Inverse optimal control is the problem of computing the (unknown) parameters of an optimal control problem's objective function such that given state and control trajectories are optimal (see \cite{Molloy2018,Keshavarz2011,Johnson2013,Pauwels2014} and references therein).
In contrast, the standard problem of (forward) optimal control involves finding optimal state and control trajectories given complete knowledge of the objective function.
The solution of (forward) optimal control problems with state and/or control constraints has received considerable recent attention, and a variety of efficient solution techniques now exist including the exact penalty method \cite{Li2011,Li2013} and the constraint transcription method \cite{Liu2014,Liu2017,Li2011a} (see also \cite{Yang2016} for a summary of the implementation and use of these and other techniques).
In these methods, constrained (forward) optimal control problems that are difficult (or intractable) to solve analytically are solved by employing novel control parameterisation schemes that parameterise the optimal controls as combinations of basis functions.
Despite the recent success in solving constrained (forward) optimal control problems, the inverse optimal control problem has received considerably less attention in settings where the states and controls may be subject to constraints and the horizon is unknown and potentially infinite.


Under the assumption that the horizon is known and finite, methods of inverse optimal control have been proposed on the basis of bilevel optimisation \cite{Mombaur2010}, Karush-Kuhn-Tucker (KKT) conditions \cite{Keshavarz2011,Puydupin2012,Jin2018}, Pontryagin's minimum principle \cite{Molloy2018,Molloy2016,Johnson2013}, and the Hamilton-Jacobi-Bellman equation \cite{Pauwels2014}.
Few of these methods are directly applicable in discrete-time settings when the given trajectories contain active control constraints.
For example, neither the recent methods nor performance guarantees of \cite{Molloy2018} and \cite{Jin2018} are applicable when the given trajectories contain active control constraints.
Furthermore, the majority of these finite-horizon inverse optimal control methods (including those in \cite{Molloy2018} and \cite{Jin2018}) store and process the given trajectories in batches or in their entirety.
They therefore lack efficient online implementations and their memory and computational complexities increase with the length of the given trajectories.

Methods of inverse optimal control have also been proposed under the assumption that the horizon is infinite \cite{Priess2015,Keshavarz2011,Molloy2018b,Boyd1994,Kalman1964}.
As in the finite-horizon case, few (if any) of these inverse methods are applicable when the given trajectories contain active control constraints.
Indeed, most existing infinite-horizon inverse methods are limited to unconstrained linear systems with quadratic objective functions \cite{Priess2015,Keshavarz2011,Molloy2018b,Boyd1994,Kalman1964}.
For example, the infinite-horizon method of \cite{Priess2015} is wholly reliant on this linear-quadratic structure since it involves computing a feedback gain matrix and then computing the objective-function parameters by solving a system of linear matrix inequalities (see \cite[Section 10.6]{Boyd1994} and references therein).
Similarly, the approach of \cite{Keshavarz2011} assumes a quadratic form of the objective function and relies on the very restrictive assumption that the stage function of the optimal control problem is known.
Despite these efforts and our recent proposal of a method of infinite-horizon inverse optimal control for discrete-time unconstrained nonlinear systems in \cite{Molloy2018b}, the problem of control-constrained inverse optimal control remains largely unresolved in both finite and infinite horizon settings.

The key contribution of this paper is the proposal of a novel method of online discrete-time inverse optimal control that computes objective-function parameters from trajectories with control constraints.
A secondary contribution of this paper is the establishment of conditions under which our proposed online method is guaranteed to compute unique objective-function parameters.
We develop our method and performance guarantees by establishing a new combined discrete-time minimum principle for both finite and infinite horizon optimal control problems that involves a forward recursion for the costates (rather than the backward recursions present in prior art, cf.~\cite{Molloy2018,Blot2000,Bertsekas2005,Goodwin2006}).
By exploiting this combined minimum principle, our method and performance guarantees are applicable to both finite and infinite horizon problems with constrained controls without requiring explicit knowledge of the horizon.
In contrast, the recent treatments of discrete-time inverse optimal control in \cite{Molloy2018, Jin2018, Molloy2018b} are specialised to either finite or infinite horizon settings and are only applicable to trajectories without control constraints.
Thus, in the finite-horizon setting, our method contrasts with those of \cite{Molloy2018} and \cite{Jin2018} by handling trajectories with constrained controls, having an efficient online implementation, and not requiring prior knowledge of the horizon.
In the infinite-horizon setting, in contrast to the method of \cite{Molloy2018b}, our method handles trajectories with constrained controls.


This paper is structured as follows.
In Section \ref{sec:problem}, we pose online inverse optimal control.
In Section \ref{sec:ioc}, we develop a combined minimum principle for both finite and infinite horizon discrete-time optimal control problems and propose our novel method of online inverse optimal control.
In Section \ref{sec:performance} we establish performance guarantees for our method.
We present an illustrative example in Section \ref{sec:examples} and provide conclusions in Section \ref{sec:conclusion}.

\section{Problem Formulation}
\label{sec:problem}
Let us consider the discrete-time deterministic system
\begin{equation}
 \label{eq:dynamics}
 x_{k+1}
 = f_k \left( x_k, u_k \right),
 \quad  x_0 \in \mathbb{R}^n
\end{equation}
for $k \geq 0$ where $f_k : \mathbb{R}^n \times \mathcal{U} \mapsto \mathbb{R}^n$ are continuously differentiable (possibly nonlinear) functions, $x_k \in \mathbb{R}^n$ are state vectors, and $u_k \in \mathcal{U}$ are (potentially multidimensional) control variables belonging to the closed and convex constraint set $\mathcal{U} \subset \mathbb{R}^m$.
Let us define the objective function
\begin{equation}
 \label{eq:ocCost}
 V \left( \x{0}{K}, \u{0}{K}, \theta \right)
 \triangleq \sum_{k = 0}^{K} \theta' L_k\left(x_k, u_k \right)
\end{equation}
with the possibly infinite horizon $K > 0$ where $\theta \in \Theta$ is a time-invariant parameter vector from the parameter set $\Theta \subset \mathbb{R}^N$, and $L_k : \mathbb{R}^n \times \mathbb{R}^m \mapsto \mathbb{R}^N$ for $k \geq 0$ are basis functions that are continuously differentiable in both of their arguments.
We shall use $^\prime$ to denote the transpose operator, $\x{0}{K}$ to denote the state sequence $\{x_k : 0 \leq k \leq K\}$ and $\u{0}{K}$ to denote the control sequence $\{u_k : 0 \leq k \leq K\}$.
In the (well-posed) discrete-time optimal control problem, we solve
\begin{align}
 \label{eq:ocProblem}
 \begin{aligned}
  &\underset{\u{0}{K}}{{\inf}}  & & V \left( \x{0}{K}, \u{0}{K}, \theta \right) < \infty\\
  &\mathrm{s.t.} & & x_{k+1} = f_k (x_k, u_k), \quad k \geq 0 \\
  & & & u_k \in \mathcal{U}, \quad k \geq 0\\
  & & & x_{0} \in \mathbb{R}^n
 \end{aligned}
\end{align}
for the optimal state $\x{0}{K}$ and control $\u{0}{K}$ trajectories given knowledge of the possibly infinite horizon $K$, the dynamics $f_k$, the constraint set $\mathcal{U}$, the time-invariant parameter vector $\theta$, and the basis functions $L_k$.

In this paper, we consider the problem of inverse optimal control in which we seek to compute parameter vector $\theta \in \Theta$ of the objective function \eqref{eq:ocCost} such that a (possibly infinite) pair of state and control trajectories $\x{0}{K}$ and $\u{0}{K}$ constitute an optimal solution to the optimal control problem \eqref{eq:ocProblem}.
We shall specifically consider a novel \emph{online} inverse optimal control problem in which we seek to compute the parameter vector $\theta$ from a sequence of state and control pairs $(x_k,u_k)$ drawn from the (possibly infinite) trajectories $\x{0}{K}$ and $\u{0}{K}$ without storing and processing the pairs in batches.
In this inverse optimal control problem, we assume that we have knowledge of the dynamics $f_k$, basis functions $L_k$, and constraint set $\mathcal{U}$.
We note that in contrast to previous formulations of discrete-time inverse optimal control (cf.~\cite{Molloy2018}), our online inverse optimal control problem assumes no prior knowledge of the (possibly infinite) horizon $K$ and prohibits the storage of the trajectories $\x{0}{K}$ and $\u{0}{K}$.


On occasion in this paper, we shall make use of the following assumption to differentiate between cases where the trajectories $\x{0}{K}$ and $\u{0}{K}$ constitute a solution to the optimal control problem \eqref{eq:ocProblem} for some $\theta = \theta^* \in \Theta$, and cases where they do not constitute a solution to \eqref{eq:ocProblem} for any $\theta = \theta^* \in \Theta$.

\begin{assumption}[Forward Optimality]
 \label{assumption:trueParameters}
 The trajectories $\x{0}{K}$ and $\u{0}{K}$ constitute a solution to \eqref{eq:ocProblem} with dynamics $f_k$, basis functions $L_k$, constraint set $\mathcal{U}$, unknown (possibly infinite) horizon $K$, and unknown unique objective-function parameter vector $\theta = \theta^* \in \Theta$.
\end{assumption}

In this paper, we also seek to investigate conditions under which our inverse optimal control problem has a unique solution (especially under Assumption 1).
As a first step towards establishing these conditions, we note that scaling the objective function $V$ of the optimal control problem \eqref{eq:ocProblem} by any $r > 0$ does not change the nature of the optimising trajectories $\x{0}{K}$ and $\u{0}{K}$ but does scale the minimum value of the objective function $V$.
Thus, an immediate condition necessary (though not sufficient) for our inverse optimal control problem to possess a unique solution is that parameter set $\Theta$ must not contain $\theta = \theta^*$ and $\theta = r\theta^*$ for any $r > 0$ and any $\theta^*$.
In this paper we follow existing approaches (cf.~\cite{Molloy2018,Molloy2016,Molloy2018b}), and consider the parameter set to be of the form $\Theta \triangleq \{ \theta \in \mathbb{R}^N : \theta^1 = a\}$ for some scalar $a > 0$.
We note that there is no loss of generality with this choice of parameter set and we expect results analogous to those of this paper to hold when the parameter set is instead constructed as the fixed-normalisation set $\Theta = \{ \theta \in \mathbb{R}^{N} : \|\theta\| = a\}$ as in \cite{Albrecht2011} (see also \cite{Molloy2018b} for a comparison of infinite-horizon inverse optimal control results with fixed-element and fixed-normalisation parameter sets).

\section{Online Inverse Optimal Control}
\label{sec:ioc}

In this section, we exploit minimum principles for both finite-horizon and infinite-horizon discrete-time optimal control problems to propose our novel method of online inverse optimal control.

\subsection{Finite and Infinite Horizon Minimum Principles}

To present the discrete-time minimum principles that we shall exploit, let us define the Hamiltonian function associated with the optimal control problem \eqref{eq:ocProblem} as
\begin{equation}
 \label{eq:hamDef}
 H_k \left( x_k, u_k, \lambda_{k+1}, \theta \right)
 \triangleq \theta' L_k\left( x_k, u_k \right) + \lambda_{k+1}' f_k \left( x_k, u_k \right)
\end{equation}
where $\lambda_k \in \mathbb{R}^n$ for $k \geq 0$ are costate (or adjoint) vectors.
Let us also define $\nabla_{x} H_k \left( x_k, u_k, \lambda_{k+1}, \theta \right) \in \mathbb{R}^n$ and $\nabla_{u} H_k \left( x_k, u_k, \lambda_{k+1}, \theta \right) \in \mathbb{R}^m$ as the column vectors of partial derivatives of the Hamiltonian with respect to $x_k$ and $u_k$, respectively, and evaluated at $x_k$, $u_k$, $\lambda_{k+1}$, and $\theta$.
We shall similarly use $\nabla_x f_k \in \mathbb{R}^{n \times n}$, $\nabla_u f_k \in \mathbb{R}^{m \times n}$, $\nabla_x L_k \in \mathbb{R}^{n \times N}$, and $\nabla_u L_k \in \mathbb{R}^{m \times N}$ to denote the matrices of partial derivatives of $f_k$ and $L_k$.
We also require the following assumption.

\begin{assumption}[Jacobian Invertibility]
 \label{assumption:invertable}
 The derivative matrix of the dynamics $\nabla_x f_k$ at $(x_k, u_k)$ is invertible for all $k \geq 0$.
\end{assumption}

Assumption \ref{assumption:invertable} is potentially restrictive; for example, it corresponds to requiring the invertibility of the state transition matrix in linear systems.
However it has previously been used to establish both finite and infinite horizon discrete-time minimum principles (cf.~\cite{Blot2000} and \cite[Theorem 3.3.1]{Goodwin2006}).
We shall use Assumption \ref{assumption:invertable} to combine the finite-horizon minimum principle of \cite[Proposition 3.3.2]{Bertsekas2005} with the infinite-horizon minimum principle of \cite[Theorem 2]{Blot2000}.
Before we present this combined minimum principle, let us introduce the following definition.

\begin{definition}[Inactive Constraint Times]
Given the controls $u_k$ for $k \geq 0$, we shall define the \emph{inactive constraint times} up to and including some time $\ell \geq 0$ as the set of times
\begin{equation*}
    \mathcal{K}_\ell 
    \triangleq \{0 \leq k \leq \ell : u_k \in \interior \mathcal{U} \}
\end{equation*}
where $u_k \in \interior \mathcal{U}$ denotes that the control $u_k$ is in the interior (i.e., not on the boundary) of the control constraint set $\mathcal{U}$.
\end{definition}

We now present our combined finite and infinite horizon discrete-time minimum principle.

\begin{lemma}
 \label{lemma:minP}
 If Assumptions \ref{assumption:trueParameters} and \ref{assumption:invertable} hold so that $\x{0}{K}$ and $\u{0}{K}$ constitute a solution to \eqref{eq:ocProblem} with $\theta = \theta^* \in \Theta$ and potentially infinite $K > 0$, then
\begin{equation}
 \label{eq:backwardsInduction}
 \lambda_k
 = \nabla_{x} H_k \left( x_k, u_k, \lambda_{k+1}, \theta \right)
\end{equation}
for all $0 \leq k \leq K$ with $\lambda_{K+1} = 0$ if $K < \infty$ and $\lambda_{K+1}$ undefined if $K = \infty$.
Furthermore,
\begin{equation}
 \label{eq:minPrinciple}
 \nabla_u H_k \left( x_k, u_k, \lambda_{k+1}, \theta \right)
 = 0
\end{equation}
for all $k \in \mathcal{K}_K$ where $\mathcal{K}_K$ are the \emph{inactive constraint times} up to and including time $K$.
\end{lemma}
\begin{proof}
    In the case $0 < K < \infty$, \cite[Proposition 3.3.2]{Bertsekas2005} establishes that $\lambda_k$ satisfies \eqref{eq:backwardsInduction} for $0 \leq k \leq K$ with $\lambda_{K+1} = 0$, and that
    \begin{equation}
        \label{eq:variationalIneq}
        \nabla_{u} H_k \left( x_k, u_k, \lambda_{k+1}, \theta \right)' \left( \bar{u} - u_k \right) \geq 0
    \end{equation}
    for all $\bar{u} \in \mathcal{U}$ and all $0 \leq k \leq K$.
    The variational inequality \eqref{eq:variationalIneq} simplifies to \eqref{eq:minPrinciple} at times $k \in \mathcal{K}_K$ proving the lemma assertion when $K$ is finite.
    
    In the case $K = \infty$, \cite[Theorem 2]{Blot2000} under Assumption \ref{assumption:invertable} establishes that $\lambda_k$ satisfies \eqref{eq:backwardsInduction} for $k \geq 0$ without a defined terminal or initial condition, and $u_k$ satisfies the variational inequality \eqref{eq:variationalIneq} for all $\bar{u} \in \mathcal{U}$ and all $k \geq 0$.
    Again, the variational inequality \eqref{eq:variationalIneq} simplifies to \eqref{eq:minPrinciple} at times $k \in \mathcal{K}_K$ proving the lemma assertion when $K$ is infinite and completing the proof.
\end{proof}

Lemma \ref{lemma:minP} describes the properties of the costates and the gradients of the Hamiltonian when the state $\x{0}{K}$ and control $\u{0}{K}$ trajectories constitute a solution to the optimal control problem \eqref{eq:ocProblem} with $\theta = \theta^* \in \Theta$ for any (possibly infinite) horizon $K > 0$.
We note that the terminal boundary condition for $\lambda_{K+1}$ is only defined in the case of a finite horizon $K < \infty$, and no boundary or initial conditions are imposed on $\lambda_k$ in the case of an infinite horizon $K = \infty$ (consistent with the infinite-horizon minimum principle of \cite{Blot2000}).
In the follow theorem, we shall omit the terminal boundary condition $\lambda_{K+1} = 0$ when $K < \infty$ and use Assumption \ref{assumption:invertable} in order to convert the costate backward recursion \eqref{eq:backwardsInduction} to a forward recursion.
We will later exploit our forward recursion to propose our method of online inverse optimal control.

\begin{theorem}
\label{theorem:linearSystem}
 If Assumptions \ref{assumption:trueParameters} and \ref{assumption:invertable} hold so that \eqref{eq:ocProblem} is solved by $\x{0}{K}$ and $\u{0}{K}$ with $\theta = \theta^* \in \Theta$ and potentially infinite $K > 0$, then
  \begin{align}
 \label{eq:hamSystemeasier}
 F_k \mathcal{G}_k \alpha &= 0
\end{align}
for all $k \in \mathcal{K}_K$ where
$\alpha \triangleq [
  \theta' \;
  \lambda_0'
 ]'
$ and
\begin{align}
\label{eq:fMatrix}
 F_k
 \triangleq \begin{bmatrix}
  \nabla_u L_k && \nabla_u f_k
 \end{bmatrix}.
 \end{align}
 Here, $\mathcal{G}_k \triangleq \prod_{\ell = 0}^{k} G_{\ell}$ is given by the forward recursion
 \begin{align}
  \label{eq:gRecursion}
  \mathcal{G}_k
  &= G_{k} \times \mathcal{G}_{k-1}
 \end{align}
 for $k \geq 1$ with $\mathcal{G}_{0} = G_0$ and
 \begin{align}
 \label{eq:gMatrix}
 G_k 
 \triangleq
 \begin{bmatrix}
  I && 0 \\
  - \nabla_x f_k^{-1} \nabla_x L_k && \nabla_x f_k^{-1}
 \end{bmatrix} \in \mathbb{R}^{(n + N) \times (n + N)}.
 \end{align}
\end{theorem}
\begin{proof}
 The definition of the Hamiltonian \eqref{eq:hamDef} combined with the backward recursion \eqref{eq:backwardsInduction} from Lemma \ref{lemma:minP} holding under Assumption \ref{assumption:invertable} implies that
 \begin{align}
 \label{eq:forwardRecursion}
 \lambda_{k+1}
 &= \nabla_x f_k^{-1} \lambda_{k} - \nabla_x f_k^{-1} \nabla_x L_k \theta
 \end{align}
 for all $k \geq 0$ where we have noted that $\nabla_x f_k$ is invertible under Assumption \ref{assumption:invertable}.
 By defining $z_k \triangleq [ \theta' \;  \lambda_k']'$ for $k \geq 0$ and recalling the definitions of $G_k$ and $\mathcal{G}_k$, \eqref{eq:forwardRecursion} may be rewritten as the forward recursion
 \begin{align}\label{eq:recursive}
 z_{k+1}
 = G_k z_k
 = \mathcal{G}_{k} \alpha
 \end{align}
 for $k \geq 0$ where we note that $z_0 = \alpha$.
 Similarly, applying the definition of the Hamiltonian \eqref{eq:hamDef} to \eqref{eq:minPrinciple} under Assumption \ref{assumption:invertable} gives
 \begin{align*}
 0 
 &= \nabla_u L_k \theta + \nabla_u f_k \lambda_{k+1}\\
 &= F_k \mathcal{G}_{k} \alpha
 \end{align*}
 for $k \in \mathcal{K}_K$ where the last line follows from \eqref{eq:recursive}.
 The theorem assertion follows and the proof is complete.
\end{proof}

The matrix equation \eqref{eq:hamSystemeasier} summarises both the costate \eqref{eq:backwardsInduction} and Hamiltonian-gradient \eqref{eq:minPrinciple} conditions of Lemma \ref{lemma:minP}.
By rewriting the costate backward recursion \eqref{eq:backwardsInduction} as a forward recursion, we have eliminated the costate vectors $\lambda_k$ for $k \geq 1$ from \eqref{eq:hamSystemeasier}.
When Assumption \ref{assumption:trueParameters} holds, we may thus, in principle, solve \eqref{eq:hamSystemeasier} at any time $k \in \mathcal{K}_K$ for the vector $\alpha$ which will yield values of the parameter vector $\theta$ and initial costates $\lambda_0$.
However, in practice the matrices $F_k$ and $\mathcal{G}_k$ may be rank deficient and the equality in \eqref{eq:hamSystemeasier} may not hold exactly due to violation of Assumption \ref{assumption:trueParameters}; for example, the given trajectories may not be optimal for any $\theta \in \Theta$ due to misspecified dynamics or basis functions.
To handle these situations, we shall next propose an inverse optimal control method by considering sums of squared residuals $\| F_k \mathcal{G}_k \alpha \|^2$.

\subsection{Proposed Online Inverse Optimal Control Method}
To propose our online inverse optimal control method, let us consider the \emph{inactive constraint times} up to and including time $k \geq 0$, namely, $\mathcal{K}_k$.
Under Assumption \ref{assumption:invertable}, let us also define
\begin{align}\notag
 J_k \left( \alpha \right)
 &\triangleq \sum_{\ell \in \mathcal{K}_k} \left\| F_\ell \mathcal{G}_\ell \alpha \right\|^2\\ \label{eq:ocost}
 &= \alpha' \mathcal{Q}_k \alpha 
\end{align}
as the sum of squared residuals of \eqref{eq:hamSystemeasier} where
\begin{align*}
    \mathcal{Q}_k 
    &\triangleq \sum_{\ell \in \mathcal{K}_k} \left( F_\ell \mathcal{G}_\ell \right)'\left(F_\ell \mathcal{G}_\ell \right)
\end{align*}
is a symmetric positive semidefinite matrix.
Our proposed method of online inverse optimal control is then to find vectors $\hat{\alpha}_k$ at each time $k \geq 0$ that solve the optimisation problem
\begin{align}
 \label{eq:method}
 \begin{aligned}
  &\inf_{\alpha}  & & J_k \left( \alpha \right)
  &\mathrm{s.t.} & & \mathcal{I}\alpha \in \Theta
 \end{aligned}
\end{align}
where $\mathcal{I} \triangleq [I \; 0 ] \in \mathbb{R}^{N \times (N + n)}$.
The objective-function parameter vector $\hat{\theta}_k$ and initial costates $\hat{\lambda}_0$ computed by our method are then given by $\hat{\theta}_k = \mathcal{I} \hat{\alpha}_k$ and $\hat{\lambda}_0 = \bar{\mathcal{I}} \hat{\alpha}_k$ where $\bar{\mathcal{I}} \triangleq [ 0 \; I ] \in \mathbb{R}^{n \times (N + n)}$.

Our method \eqref{eq:method} has an online form in the sense that it can process the pairs $(x_k,u_k)$ for $k \geq 0$ sequentially since $\mathcal{Q}_k$ is given by the recursion
\begin{align}
   \label{eq:qRecursion}
    \mathcal{Q}_k
    &= \begin{cases}
    	\mathcal{Q}_{k-1} + \left( F_k \mathcal{G}_k \right)'\left(F_k \mathcal{G}_k \right) & \text{if } u_k \in \interior \, \mathcal{U},\\
	\mathcal{Q}_{k-1} & \text{otherwise}
	\end{cases}
\end{align}
for $k \geq 0$ where $\mathcal{Q}_{-1} \triangleq 0$
and $\mathcal{G}_k$ is given by the recursion \eqref{eq:gRecursion}.
Furthermore, the dimensionality of the optimisation in our method is $N + n$ and does not grow with time.
In contrast, the dimensionality of the optimisation problems in existing minimum principle and KKT methods of inverse optimal control grow linearly with the length of the trajectories considered since, for example, they involve optimisation over the entire costate trajectory $\lambda_0, \lambda_1, \ldots, \lambda_k$ \cite{Molloy2018,Keshavarz2011,Puydupin2012,Jin2018}.
The key to the constant dimensionality and efficient online form of our method is the forward-recursive expression of the finite and infinite horizon minimum principles established in Theorem \ref{theorem:linearSystem}.

By minimising the residual cost function \eqref{eq:ocost} over $\alpha$, our method \eqref{eq:method} computes a parameter vector $\hat{\theta}_k$ and initial costates $\hat{\lambda}_0$ that minimise the violation of the minimum principle conditions \eqref{eq:hamSystemeasier}.
If Assumption \ref{assumption:trueParameters} holds so that $\x{0}{K}$ and $\u{0}{K}$ are a solution to the optimal control problem \eqref{eq:ocProblem} for $\theta = \theta^* \in \Theta$, then $\hat{\alpha}_k = [\theta^{*\prime} \; \lambda_0']'$ will be one (possibly nonunique) solution to \eqref{eq:method} for all $k \geq 0$.
If Assumption \ref{assumption:trueParameters} does not hold so that $\x{0}{K}$ and $\u{0}{K}$ are suboptimal under \eqref{eq:ocCost} for all $\theta \in \Theta$ (e.g., due to noise or misspecified basis functions and dynamics), then our method \eqref{eq:method} will yield parameters $\hat{\theta}_k$ that minimise the extent to which the minimum principle conditions of \eqref{eq:hamSystemeasier} are violated.
The extent to which the minimum principle conditions of \eqref{eq:hamSystemeasier} are violated can be determined online since $J_k ( \hat{\alpha}_k) = \hat{\alpha}_k' \mathcal{Q}_k \hat{\alpha}_k$ will be equal to zero when Assumption \ref{assumption:trueParameters} holds and greater than zero when it does not.
We note that if $J_k ( \hat{\alpha}_k)$ is large, the method is unlikely to yield useful parameters without modification of the dynamics or basis functions, or preprocessing the trajectories to remove noise as in \cite{Johnson2013}.

Methods for computing parameter vectors that minimises the extent to which the minimum principle conditions of \eqref{eq:hamSystemeasier} are violated have been previously proposed for problems with known finite horizons $K < \infty$ (cf.~\cite{Molloy2018,Keshavarz2011,Puydupin2012}).
The key novelty of our method \eqref{eq:method} is that it exploits our novel reformulation of the minimum principle conditions in Theorem \ref{theorem:linearSystem} to yield parameter vectors online without any prior knowledge of the (potentially infinite) horizon $K$ and without storing or processing batches of states and controls.
Furthermore, unlike the existing methods of \cite{Molloy2018,Keshavarz2011,Puydupin2012}, our method handles state and control trajectories with active control constraints.
We shall next establish conditions under which our method computes a unique parameter vector, and we will further describe its online implementation.

\section{Performance Guarantees and Online Implementation}
\label{sec:performance}

In this section, we present our main result guaranteeing the uniqueness of solutions to our method \eqref{eq:method}. We also describe its online implementation.

\subsection{Performance Guarantees}

To establish our main performance result, let us define $\bar{\mathcal{Q}}_k \in \mathbb{R}^{(n + N -1) \times (n + N -1)}$ as 
the principal submatrix of $\mathcal{Q}_k$ formed by removing its first row and first column, and let us also define
$
q_k \in \mathbb{R}^{n + N - 1}
$ as the first column of $\mathcal{Q}_k$ without its first element.
We now present our main result and performance guarantee.

\begin{theorem}
\label{theorem:conElement}
Consider $(x_k,u_k)$ for $k \geq 0$, suppose that Assumption \ref{assumption:invertable} holds, and let $\Theta = \{ \theta \in \mathbb{R}^N : \theta^1 = a\}$ for some $a > 0$.
For any $k \geq 0$, if $\bar{\mathcal{Q}}_k$ has full rank then the unique solution to \eqref{eq:method} is
\begin{align}
 \label{eq:alphaConElement}
 \hat{\alpha}_k
 &=a \begin{bmatrix}
        1\\ 
        -\bar{\mathcal{Q}}_k^{-1}q_k
     \end{bmatrix}.
\end{align}
If, in addition, Assumption \ref{assumption:trueParameters} holds and there exists an $r > 0$ such that $r \theta^* \in \Theta$, then the unique solution to \eqref{eq:method} is
\begin{align*}
 \hat{\alpha}_k
 =a \begin{bmatrix} 
    1\\ 
     -\bar{\mathcal{Q}}_k^{-1}q_k 
    \end{bmatrix}
 = r\begin{bmatrix}
 \theta^*\\
 \lambda_0
 \end{bmatrix}.
\end{align*}
\end{theorem}
\begin{proof}
Under Assumption \ref{assumption:invertable} and given $\Theta$, the Lagrangian function for \eqref{eq:method} for any $k \geq 0$ is 
\begin{align*}
 \mathcal{L}(\alpha,p) = \alpha' \mathcal{Q}_k \alpha + p (\alpha^1 - a)
\end{align*}
with the Lagrange multiplier $p \in \mathbb{R}$.
Noting that $\mathcal{Q}_k$ is symmetric, and letting $c \triangleq p/2$, we have the derivatives 
\begin{align*}
 \dfrac{d\mathcal{L}(\alpha,p)}{dp}
 &= \alpha^1 - a
\text{ and }
    \dfrac{d\mathcal{L}(\alpha,p)}{d\alpha} = 2 \mathcal{Q}_k \alpha + 2c e_1
\end{align*}
where $e_1 \in \mathbb{R}^{N + n}$ is an indicator vector with $1$ in its first component and zeros elsewhere.
Following the method of Lagrange multipliers, setting $d\mathcal{L}(\hat{\alpha}_k,p)/dp = 0$ leads to $\hat{\alpha}_k^1 = a$, and setting $d\mathcal{L}(\hat{\alpha}_k,p)/d\alpha = 0$ whilst noting that $\mathcal{Q}_k$ is symmetric leads to the system
\begin{align*}
\begin{bmatrix}
 \mathcal{Q}_k^{1,1}&  q_k' \\
 q_k & \bar{\mathcal{Q}}_k
\end{bmatrix}
\begin{bmatrix}
a \\
\bar{\alpha}_k
\end{bmatrix} = -c e_1
\end{align*}
where 
$
\bar{\alpha}_k \triangleq [ \hat{\alpha}_k^2 \; \cdots \; \hat{\alpha}_k^{N + n}]'
$ and $\mathcal{Q}_k^{1,1}$ is the first element of $\mathcal{Q}_k$.
Since $a$ is known, we may discard the first row of both sides and perform straightforward matrix manipulations to obtain
the equation $
\bar{\mathcal{Q}}_k \bar{\alpha}_k+ a q_k = 0.
$
Since $\bar{\mathcal{Q}}_k $ is invertible under the theorem conditions, we have
$
\bar{\alpha}_k 
= - a \bar{\mathcal{Q}}_k^{-1}  q_k
$ and the first theorem result follows.

To prove the second theorem result we note that \eqref{eq:hamSystemeasier} established in Theorem \ref{theorem:linearSystem} under Assumptions \ref{assumption:trueParameters} \& \ref{assumption:invertable} holds with $\alpha = \alpha^* = r [\theta^{*\prime} \; \lambda_0']'$ for all $r > 0$.
Thus, under Assumptions \ref{assumption:trueParameters} \& \ref{assumption:invertable}, $\hat{\alpha}  = \alpha^*$ is a solution to \eqref{eq:method}.
Recalling the first theorem assertion, $\hat{\alpha} = \alpha^*$ will be the unique solution when $\bar{\mathcal{Q}}_k$ is full rank and $r > 0$ is such that $\theta^* \in \Theta$.
The proof is complete.
\end{proof}

The first assertion of Theorem \ref{theorem:conElement} establishes that by constraining $\hat{\theta}^1 = a$, our method \eqref{eq:method} is guaranteed to compute a unique parameter vector given by 
\begin{align}
    \label{eq:parameters}
    \hat{\theta}_k
    &= a \mathcal{I} \begin{bmatrix} 
    1\\ 
     -\bar{\mathcal{Q}}_k^{-1}q_k 
    \end{bmatrix}
\end{align}
when the submatrix $\bar{\mathcal{Q}}_k$ has full rank.
The second assertion of Theorem \ref{theorem:conElement} establishes that when Assumption \ref{assumption:trueParameters} holds, the unique parameter vector will correspond to a scaled version of the true unknown parameter vector $\theta^*$ provided that the true parameter vector can be re-scaled by $r > 0$ to belong to the set $\Theta$.
This scaling condition is typically not restrictive but may require the permutation of the basis functions $L_k$ (in practice via trial and error) to avoid the first element of $\theta^*$ being zero.

The first assertion of Theorem \ref{theorem:conElement} holds without Assumption \ref{assumption:trueParameters} and its rank condition involving $\bar{\mathcal{Q}}_k$ is always useful for determining if the sequence of state and control pairs $(x_\ell,u_\ell)$ for $0 \leq \ell \leq k$ provide sufficient information for the parameter vector and initial costate vector to be computed uniquely with our method (regardless of Assumption \ref{assumption:trueParameters}).
The rank condition therefore fulfils a role analogous to the persistence of excitation conditions that appear in adaptive control and parameter estimation in which data or experiments are assessed as sufficient or not to yield unambiguous parameter estimates (without considering the properties of the estimates themselves).
A similar rank condition is established for finite-horizon inverse optimal control without control constraints in \cite{Molloy2018,Molloy2016} for offline methods that involve storing and processing the full sequences $\x{0}{K}$ and $\u{0}{K}$.
Our rank condition of Theorem \ref{theorem:conElement} will fail to hold when the system dynamics or initial conditions are degenerate and lead to uninformative state and control trajectories (e.g., trajectories that are equilibria of the dynamics).
It will also fail to hold when too few state and control pairs are used to construct $\bar{\mathcal{Q}}_k$ (e.g., due to a short horizon $K$ or $\mathcal{K}_k$ having too few elements because $k$ is small or the control constraints being active too frequently).
If $\bar{\mathcal{Q}}_k$ is rank deficient, it is therefore advantageous in practice to wait for more state and control pairs to be processed before seeking to compute a unique parameter vector $\theta$.
The following proposition reinforces this intuition by showing that the rank of $\bar{\mathcal{Q}}_k$ is non-decreasing as more state and control pairs are processed.

\begin{proposition}
\label{proposition:rank}
Consider $(x_k,u_k)$ for $k \geq 0$ and suppose that Assumption \ref{assumption:invertable} holds.
Then,
\begin{align*}
    \rank \left( \bar{\mathcal{Q}}_{k-1} \right)
    \leq \rank \left( \bar{\mathcal{Q}}_k \right) 
\end{align*}
for all $k \geq 1$.
\end{proposition}
\begin{proof}
 Consider any $k \geq 1$ and note that $\bar{\mathcal{Q}}_{k}$ and $\bar{\mathcal{Q}}_{k-1}$ exist under Assumption \ref{assumption:invertable}.
 The proposition holds trivially when $u_k \not\in \interior \, \mathcal{U}$ since \eqref{eq:qRecursion} implies
 $
     \bar{\mathcal{Q}}_{k} = \bar{\mathcal{Q}}_{k-1}.
 $
 If $u_k \in \interior \, \mathcal{U}$, then \eqref{eq:qRecursion} implies that
 $
     \bar{\mathcal{Q}}_{k} = \bar{\mathcal{Q}}_{k-1} + \bar{\mathcal{F}}_k
 $
 where $\bar{\mathcal{F}}_k$ is the product matrix $(F_k\mathcal{G}_k)'(F_k\mathcal{G}_k)$ without its first row and first column. 
 By noting that $\bar{\mathcal{F}}_k$ is positive semidefinite, we have that
\begin{align*}
 v'\bar{\mathcal{Q}}_{k}v
 = v'\bar{\mathcal{Q}}_{k-1}v + v'\bar{\mathcal{F}}_kv
 \geq v'\bar{\mathcal{Q}}_{k-1}v
\end{align*}
for all $v \in \mathbb{R}^{n + N}$.
The null space of $\bar{\mathcal{Q}}_{k}$ is thus a subset of the null space of $\bar{\mathcal{Q}}_{k-1}$, and so
$
 \nullM(\bar{\mathcal{Q}}_{k-1}) \geq \nullM(\bar{\mathcal{Q}}_{k}).
$
The rank-nullity theorem then implies that
\begin{align*}
 \rank(\bar{\mathcal{Q}}_{k-1})
 &= \rank(\bar{\mathcal{Q}}_{k}) + \nullM(\bar{\mathcal{Q}}_{k}) - \nullM(\bar{\mathcal{Q}}_{k-1})\\
 &\leq \rank(\bar{\mathcal{Q}}_{k})
\end{align*}
and the proof is complete.
\end{proof}

An important consequence of Proposition \ref{proposition:rank} is that if the rank condition of Theorem \ref{theorem:conElement} is satisfied at any time $\ell \geq 0$, then it will also be satisfied for all subsequent times $k \geq \ell$.
Theorem \ref{theorem:conElement} and Proposition \ref{proposition:rank} together therefore imply that if $\bar{\mathcal{Q}}_\ell$ has full rank for any $\ell \geq 0$, then our method \eqref{eq:method} will have unique solutions \eqref{eq:alphaConElement} for all $k \geq \ell$ under the conditions of Theorem \ref{theorem:conElement}.

\subsection{Online Implementation}

In light of Theorem \ref{theorem:conElement}, our method \eqref{eq:method} can be implemented by computing $\mathcal{Q}_k$, $\mathcal{G}_k$, and $F_k$ before solving \eqref{eq:parameters}.
If $\bar{\mathcal{Q}}_k$ is rank deficient, then we may substitute the inverse of $\bar{\mathcal{Q}}_k$ in \eqref{eq:alphaConElement} with the Moore-Penrose pseudoinverse of $\bar{\mathcal{Q}}_k$ to yield the minimum-norm solution to \eqref{eq:method} or simply set $\hat{\theta}_k = 0$.
The online implementation of our method \eqref{eq:method} is summarised in Algorithm \ref{algorithm:method}.


\begin{algorithm}[H]
\caption{Online Implementation of \eqref{eq:method}}
\label{algorithm:method}
\begin{algorithmic}[1]
\renewcommand{\algorithmicrequire}{\textbf{Input:}}
\renewcommand{\algorithmicensure}{\textbf{Output:}}
\REQUIRE States and controls $(x_k,u_k)$ for $k \geq 0$, dynamics $f_k$, basis functions $L_k$, constraint set $\mathcal{U}$, and parameter set $\Theta = \{ \theta \in \mathbb{R}^N : \theta^1 = a\}$.
\ENSURE Parameter vector $\hat{\theta}_k$ for $k \geq 1$.
 \FOR{$k = 0, 1, \ldots$}
    \STATE Receive $(x_k,u_k)$.
    \STATE Compute $G_k$ and $F_k$ with \eqref{eq:fMatrix} and \eqref{eq:gMatrix}.
    \IF {$k = 0$}
        \STATE Initialise $\mathcal{G}_0 = G_0$.
        \IF {$u_0 \in \interior \, \mathcal{U}$}
            \STATE Initialise $\mathcal{Q}_0 = (F_0\mathcal{G}_0)'(F_0\mathcal{G}_0)$.
        \ELSE
            \STATE Initialise $\mathcal{Q}_0 = 0$.
        \ENDIF
    \ELSE
        \STATE Compute $\mathcal{G}_k = G_k \times \mathcal{G}_{k-1}$.
        \IF {$u_k \in \interior \, \mathcal{U}$}
            \STATE Compute $\mathcal{Q}_k = \mathcal{Q}_{k-1} + (F_k\mathcal{G}_k)'(F_k\mathcal{G}_k)$.
        \ELSE
            \STATE Set $\mathcal{Q}_k = \mathcal{Q}_{k-1}$.
        \ENDIF
    \ENDIF
 \STATE Extract $\bar{\mathcal{Q}}_k$ and $q_k$ from $\mathcal{Q}_k$.
 \IF {$\rank (\bar{\mathcal{Q}}_k) = n + N - 1$}
    \STATE Compute unique $\hat{\theta}_k$ with \eqref{eq:parameters}.
 \ELSE
    \STATE Set $\hat{\theta}_k = 0$.
 \ENDIF
 \ENDFOR
\end{algorithmic} 
\end{algorithm}

The memory complexity of Algorithm \ref{algorithm:method} is dominated by the need to store the most recent $\mathcal{Q}_k$, $\mathcal{G}_k$, and $F_k$ which leads to a total memory complexity of $O(m(n + N) + (n + N)^2)$.
The computational complexity of Algorithm \ref{algorithm:method} is similarly dominated by the computation of $\mathcal{Q}_k$, $\mathcal{G}_k$, and the inversion of $\bar{\mathcal{Q}}_k$ in \eqref{eq:parameters} which leads to a computational complexity of $O(m(n+N)^2 + (n+N)^3)$ at each time $k$.
In contrast, the total memory complexity of the recently proposed finite-horizon inverse optimal control method of \cite{Molloy2018} is $O((m + n)NK)$ whilst its total computational complexity is $O(n^3K + (mNK)^3)$.
The horizon $K$ will typically be greater than the dimensions $m$, $n$ and $N$, and so the complexities of our method will typically be less than those of the method of \cite{Molloy2018} (and those of other methods, e.g. \cite{Keshavarz2011,Puydupin2012,Jin2018}).
Importantly from an online implementation perspective, the memory complexity of our method is independent of $K$ whilst the memory complexity of the method of \cite{Molloy2018} is linear in $K$.
The computational complexity of the method of \cite{Molloy2018} is also cubic in time $k$ whilst the total computational complexity of our method is only linear in time $k$.


\section{Simulation Examples}
\label{sec:examples}

In this section, we first illustrate our method in a simple illustrative example alongside the current state-of-the-art method of \cite{Molloy2018}.
We then consider an application-inspired example that cannot be solved with the method of \cite{Molloy2018} due to the presence of control constraints. 

\subsection{Illustrative Example}
Consider the single integrator
$
    x_{k+1}
    = x_k + u_k
$ with $x_k \in \mathbb{R}$ for $0 \leq k \leq K$ regulated with an optimal controller designed with the objective function
\begin{align*}
    V \left( \x{0}{K}, \u{0}{K}, \theta \right)
    &= \sum_{k = 0}^K (x_k)^2 + 5 (u_k)^2.
\end{align*}
The parameter vector and basis functions of this objective function are $\theta = \theta^* = [1 \; 5]'$ and $L_k = [(x_k)^2 \; (u_k)^2]$, respectively.
Thus, $\nabla_u L_k = [0 \; 2u_k]$, $\nabla_x L_k = [2x_k \; 0]$, and Assumption \ref{assumption:invertable} holds with $\nabla_x f_k^{-1} = 1$.

For the purpose of illustration, we simulated the optimal state and control trajectories with $K = 10$ and $x_0 = 10$ shown in Fig.~\ref{fig:exampleOptimal} and applied our method \eqref{eq:method} by following Algorithm \ref{algorithm:method}.
Specifically, at $k = 0$, we receive $(x_0,u_0) = (10.0,-3.58)$ and so computing $G_0$ and $F_0$, and initialising $\mathcal{G}_0$ and $\mathcal{Q}_0$, leads to
\begin{equation*}
    \bar{\mathcal{Q}}_0
    = \begin{bmatrix}
        51.2820 &  -7.1611\\
        -7.1611 &   1.0000
        \end{bmatrix}.
\end{equation*}
Here, $\bar{\mathcal{Q}}_0$ is rank deficient and so there is no unique solution to \eqref{eq:method}.
Thus, we set $\hat{\theta}_0 = 0$ and proceed to $k = 1$.
At $k = 1$, we receive $(x_1,u_1) = (6.42, -2.30)$ and so computing $G_1$, $F_1$, $\mathcal{G}_1$, and $\mathcal{Q}_1$ leads to
\begin{equation*}
    \bar{\mathcal{Q}}_1
    = \begin{bmatrix}
            72.3810  & -11.7545\\
            -11.7545 &  2.0000
        \end{bmatrix}
\end{equation*}
which is full rank.
Thus, noting that $q_k = [294.1 \; -52.8]'$, the solution of \eqref{eq:parameters} yields the unique parameter vector $\hat{\theta}_1 = [1 \; 5]'$ solving \eqref{eq:method}.
Our method \eqref{eq:method} therefore yields the unknown parameter vector $\theta^*$ online from only two pairs of states and controls, $(x_0,u_0)$ and $(x_1,u_1)$, without knowledge of the horizon $K$, and in a time of $1.2$ ms with our MATLAB implementation.
For comparison, we also processed the trajectories of Fig.~\ref{fig:exampleOptimal} with the state-of-the-art method of \cite{Molloy2018} which must process the entire trajectories $\x{0}{10}$ and $\u{0}{10}$, and requires knowledge of the horizon $K = 10$.
Whilst the method of \cite{Molloy2018} also computed the unknown parameter vector $\theta^*$, it took $3.4$ ms in our MATLAB implementation (over two times slower than our online method on this short horizon with the same computer system).

\begin{figure}
    \centering
    \includegraphics[width=0.99\columnwidth]{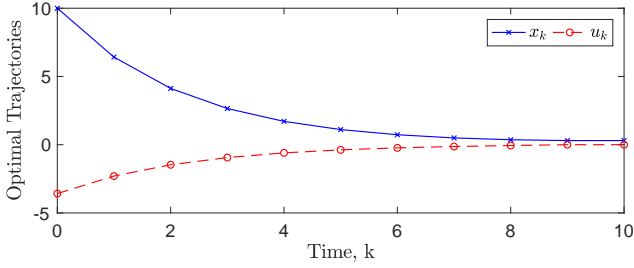}
    \caption{Simulated trajectories for illustrative example.}
    \label{fig:exampleOptimal}
\end{figure}

\subsection{Application-Inspired Example}
We now consider an example inspired by the study of how human pilots fly aircraft (cf.~\cite{Maillot2013}), and how pilot behaviours can be modelled and mimicked with optimal control problems.
We specifically consider the regulation of pitch in a fixed-wing aircraft.
Let us therefore consider a discrete-time model of aircraft pitch dynamics\footnote{http://ctms.engin.umich.edu/CTMS/index.php} 
\begin{align*}
    x_{k+1}
    &= Ax_k + Bu_k, \; x_0 \in \mathbb{R}^3
\end{align*}
for $k \geq 0$ where the three states $x_k^1, x_k^2,$ and $x_k^3$ are the angle of attack (in radians), the aircraft pitch rate (in radians per second), and the aircraft pitch angle (in radians), respectively, and
\begin{align*}
    A
    &= \begin{bmatrix}
        0.9654  & 5.4572 & 0\\
        -0.0013 & 0.9545 & 0\\
        -0.0038 & 5.5437 & 1
    \end{bmatrix}
    \text{ and }
    B
    = \begin{bmatrix}
        0.0284 & 0.0142\\
        0.0020 & 0.0010\\
        0.0056 & 0.0028
    \end{bmatrix}.
\end{align*}
Assumption \ref{assumption:invertable} holds with $\nabla_x f_k^{-1} = (A')^{-1}$.

The control input vector $u_k = [u_k^1 \; u_k^2]' \in \mathbb{R}^2$ consists of two components, the first $u_k^1$ being the deflection angle of the elevator control surface (in radians) and the second $u_k^2$ being the deflection angle of a second (smaller trim-tab) elevator control surface (in radians).
The control inputs are both constrained to the set $\mathcal{U} = \{ u = [u^1 \; u^2]' \in \mathbb{R}^2 : -\Delta \leq u^1,u^2 \leq \Delta \}$ for some constraint-magnitude $\Delta > 0$.
In an experimental setting, the controls $u_k$ would be provided by human test subjects. 
However, for the purpose of illustrating our method, we simulated the system from an initial state of $x_0 = [ 0.5 \; 0 \; 0.2]'$ and a constraint-magnitude of {$\Delta = 0.09$} after it was regulated with an optimal controller designed with
\begin{align*}
    &V \left( \x{0}{K}, \u{0}{K}, \theta \right)\\
    &\quad= \sum_{k = 0}^K x_k' \begin{bmatrix}
    1 && 0 && 0\\
    0 && 4 && 0\\
    0 && 0 && 2
    \end{bmatrix} x_k +  u_k' \begin{bmatrix}
    3 && 0\\
    0 && 6
    \end{bmatrix} u_k.
\end{align*}
Thus, our aim in this example is to recover the parameter vector $\theta = \theta^* = [1 \; 4 \; 2 \; 3 \; 6]'$ (the diagonal elements of the state and control weighting matrices) without knowledge of the horizon (which we simulated as {$K = 250$}).
The simulated (optimal) state and control trajectories are shown in Fig.~\ref{fig:optimal} for $k \leq 100$.
We see that the control constraints are active in the time interval {$k \in [6,33]$}.

We applied our method to the trajectories in Fig.~\ref{fig:optimal} using Algorithm \ref{algorithm:method}.
The unique parameter vector computed by our method for $k \leq 100$ are shown in Fig.~\ref{fig:computed}, and these correspond to the true parameter vector $\theta^*$ for {$k \geq 35$}.
Prior to {$k = 35$}, the parameter vector are not computed due to $\bar{\mathcal{Q}}_k$ being singular (or numerically close to singular) and the control constraints being active for {$k \in [6,33]$}.

\begin{figure}[!t]
 \centering
 \includegraphics[width=0.99\columnwidth]{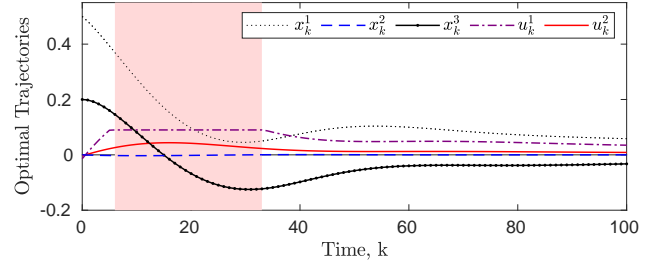}
 \caption{Simulated trajectories for application-inspired example. The control constraints are active in the shaded region.}
 \label{fig:optimal}
\end{figure}

\begin{figure}[!t]
 \centering
 \includegraphics[width=0.99\columnwidth]{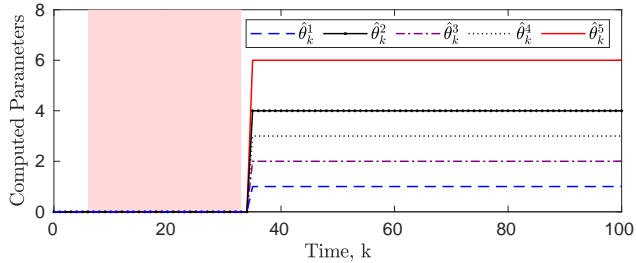}
 \caption{Objective-function parameter vector $\hat{\theta}_k$ calculated from the trajectories of Fig.~\ref{fig:optimal} using our method \eqref{eq:method}. The true parameter vector is $\theta = \theta^* = [1 \; 4 \; 2 \; 3 \; 6]'$. The control constraints are active in the shaded region.}
 \label{fig:computed}
\end{figure}

To study the impact of control constraints on the number of time steps before unique parameter vector can be computed with our method, we simulated optimal trajectories from an initial state of $x_0 = [0.5 \; 0 \; 0.2]'$ with a horizon of  {$K = 250$} and constraint magnitudes of between  {$\Delta = 0.07$ and $\Delta = 0.11$}.
We applied our method to each of these trajectories. 
For comparison purposes, we also applied an ad-hoc version of our method in which we wait until after the constraints are active for the last time before initialising the recursions for $\mathcal{Q}_k$ and $\mathcal{G}_k$. 
That is, the ad-hoc version method does not process any states and controls when the constraints are active.
In contrast, our proposed method \eqref{eq:method} processes states and controls in the recursion \eqref{eq:gRecursion} for $\mathcal{G}_k$ but not in the recursion \eqref{eq:qRecursion} for $\mathcal{Q}_k$ when the constraints are active.

Fig.~\ref{fig:delay} reports the first time at which unique parameter vector can be computed with both methods versus the duration of time the control constraints were active.
The duration of time that the control constraints are active corresponds directly to the constraint magnitude (i.e., the constraints are active for $19$ time steps when $\Delta = 0.11$ compared to $40$ time steps when $\Delta = 0.07$).
From Fig.~\ref{fig:delay}, we see that the time required by both methods to compute unique parameter vector increases with the length of time the constraints are active.
However, our proposed method \eqref{eq:method} uniformly computes the unique parameter vector in fewer time steps than the ad-hoc method.
The processing of the state and control pairs in our proposed method \eqref{eq:method} with the recursion for $\mathcal{G}_k$ while the control constraints are active is therefore advantageous (despite our method not computing new values of $\mathcal{Q}_k$ or $\hat{\theta}_k$ when the constraints are active).

\begin{figure}[!t]
 \centering
 \includegraphics[width=1\columnwidth]{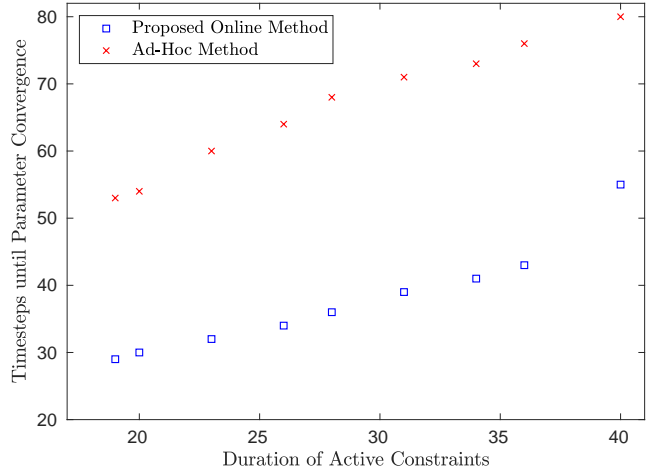}
 \caption{Time before unique objective-function parameter vector is computed versus duration of time the control constraints are active in our application-inspired example.}
 \label{fig:delay}
\end{figure}


\section{Conclusion}
\label{sec:conclusion}
We consider the problem of online inverse optimal control on possibly infinite horizons in discrete-time systems subject to control constraints.
We exploit both finite and infinite horizon discrete-time minimum principles to propose a novel online inverse optimal control method and to establish novel conditions under which it is guaranteed to compute a unique objective-function parameter vector.
We illustrate our method in simulations and demonstrate that it is able to compute unique parameter vectors online from trajectories with constrained controls.
%

\bibliography{Library}

\appendix
\end{document}